\documentclass[conference,final,twocolumn]{IEEEtran}
\usepackage{mathrsfs}

\usepackage{graphicx}
\usepackage{cite}
\usepackage{color}
\usepackage{colortbl}
\usepackage{psfrag}
\usepackage{subfigure}
\usepackage{amssymb}
\usepackage{epsfig}
\usepackage{pifont}
\usepackage{amsmath}
\usepackage{array}
\usepackage{multicol}
\usepackage{pifont}
\usepackage{indentfirst} 
\usepackage{amsfonts}
\usepackage{fancyhdr}
\usepackage{amscd}
\usepackage{bm}
\usepackage{extarrows}
\usepackage{multirow}

\newtheorem{theorem}{Theorem}

\newtheorem{definition}{Definition}

\newtheorem{remark}{Remark}
\newtheorem{proof}{Proof}

\makeatletter
\def\ifundefined{\@ifundefined}
\makeatother 

\makeatletter
\renewcommand{\citepunct}{,\penalty\@m\hskip.13emplus.1emminus.1em}
\renewcommand{\citedash}{\hbox{--}\penalty\@m}
\makeatother
\usepackage[bookmarks=false]{hyperref}

\begin{document}


\title{On the Degrees of Freedom of Asymmetric MIMO Interference Broadcast Channels}
\author{{Tingting Liu}, {\em Member, IEEE} and {Chenyang Yang}, {\em Senior Member, IEEE}
\thanks{T. Liu and C. Yang are with the
School of Electronics and Information Engineering, Beihang
University, Beijing 100191, China. (E-mail: ttliu@ee.buaa.edu.cn,
cyyang@buaa.edu.cn)}}

\maketitle
\begin{abstract}
In this paper, we study the degrees of freedom (DoF) of the asymmetric multi-input-multi-output interference broadcast channel (MIMO-IBC). By introducing a notion of connection pattern chain, we generalize the genie chain proposed in [11] to derive and prove the necessary condition of IA feasibility for asymmetric MIMO-IBC, which is denoted as irreducible condition. It is necessary for both linear interference alignment (IA) and asymptotic IA feasibility in MIMO-IBC with arbitrary configurations. In a special class of asymmetric two-cell MIMO-IBC, the irreducible condition is proved to be the sufficient and necessary condition for asymptotic IA feasibility, while the combination of proper condition and irreducible condition is proved to the sufficient and necessary condition for linear IA feasibility. From these conditions, we derive the information theoretic maximal DoF per user and the maximal DoF per user achieved by linear IA, and these DoFs are also the DoF per user upper-bounds of asymmetric $G$-cell MIMO-IBC with asymptotic IA and linear IA, respectively.
\end{abstract}


\setlength \arraycolsep{1pt}
\section{Introduction}
The degrees of freedom (DoF) can reflect the potential of
interference networks, which is the first-order approximation of sum
capacity in high signal-to-noise ratio regime
\cite{Jafar2007b,Jafar_3cell}. Recently, significant research
efforts have been devoted to find the information theoretic DoF for
multi-input-multi-output (MIMO) interference channel (MIMO-IC)
\cite{Jafar2007b,Jafar_3cell,Cadambe2008b,Tse_3cell_2011,Jafar2010DoF,Jafar_IAchainKcell,Yetis2010,Luo2012}
and MIMO interference broadcast channel (MIMO-IBC)
\cite{Zhuguangxi2010,TTTSP2013,TTTSP2013b}.

For a symmetric $G$-cell MIMO-IC where each base station (BS) and
each mobile station (MS) have $M$ antennas, the study in
\cite{Cadambe2008b} showed that  the information theoretic maximal
DoF per user is $M/2$, which can be achieved by asymptotic
interference alignment (IA) (i.e., with infinite time/frequency
extension). It implies that the sum DoF can linearly increase as
$G$, and the interference networks are not interference-limited
\cite{Cadambe2008b}. Encouraged by such a promising result, many
recent works strive to analyze the DoF for MIMO-IC and MIMO-IBC with
various settings and devise interference management techniques to
achieve the maximal DoF.

So far, the existing studies focus on the symmetric system where
each BS and each MS have $M$ and $N$ antennas, respectively. For a
three-cell symmetric MIMO-IC, the information theoretic maximal DoF
was obtained in \cite{Jafar_3cell}, which can be achieved by linear
IA (i.e., without any symbol extension or only with finite spatial
extension)\cite{Tse_3cell_2011}. For a $G$-cell symmetric MIMO-IC, the information
theoretic maximal DoF was only obtained for some configurations
\cite{Jafar2010DoF,Jafar_IAchainKcell}. For symmetric MIMO-IBC, the information
theoretic maximal DoF was obtained in \cite{TTTSP2013b} for
arbitrary configurations. The result indicates that the information
theoretic maximal DoF can be divided into two regions according to the ratio
of $M/N$. In the first region, the sum DoF of the system linearly
increases with the number of cells, which can be achieved by
asymptotic IA. In the second region, the DoF is a piecewise linear
function of $M$ and $N$ alternately, which can be achieved by
linear IA. Considering that asymmetric MIMO-IBC is more complex than
symmetric MIMO-IBC and expecting that the results for asymmetric
cases may be extended from symmetric ones, there are only a few
results on asymmetric systems. So far, only the \emph{proper
condition} is obtained and proved to be necessary for linear IA
feasibility in asymmetric MIMO-IC \cite{Luo2012} and asymmetric
MIMO-IBC \cite{TTTSP2013}. The following questions are still open
for asymmetric MIMO-IC and MIMO-IBC: what is the information
theoretic maximal DoF? when linear IA can achieve the information
theoretic maximal DoF?

To understand the potential of practical networks, including both homogeneous and heterogeneous
networks, we need to investigate the information theoretic maximal
DoF for asymmetric systems. Moreover, in symmetric systems, all
BSs or users have the same spatial resources. Consequently, it is
hard to know which resources of BSs or users participate in managing
interference, how the BSs or users remove the interference jointly,
and what is the impact on the DoF if some nodes increase or reduce
some resources (e.g., turning on or turning off antennas for energy
saving)?

In this paper, we investigate the DoF of the
asymmetric MIMO-IBC. By finding the difference of deriving the necessary conditions between symmetric and asymmetric systems and introducing a notion of connection pattern chain, we generalize the genie chain proposed in \cite{TTTSP2013b} into asymmetric MIMO-IBC and derive a necessary condition for linear IA and asymptotic IA feasibility, denoted as irreducible condition. From the irreducible condition and the combination of the proper condition and irreducible condition, we obtain the information theoretic DoF outer-bound and the DoF outer-bound achieved by linear IA in $G$-cell asymmetric MIMO-IBC, respectively. In addition, we also prove that these DoF bounds of a special class of two-cell MIMO-IBC can be achieved by asymptotic IA and linear IA, respectively.

\section{Necessary Conditions for Asymmetric $G$-cell MIMO-IBC}
Consider a $G$-cell MIMO system, where BS$_i$ equipped with $M_{i}$
antennas transmits to $K_{i}$ users each with
$N_{i_{1}},\cdots,N_{i_{K_{i}}}$ antennas in the $i$th cell, $i = 1,
\cdots, G$. BS$_i$ respectively transmits
$d_{i_{1}},\cdots,d_{i_{K_{i}}}$ data streams to its $K_i$ users,
then the total number of data streams in the $i$th cell is
$d_{i}=\sum_{k=1}^{K_{i}}d_{i_k}$. Assume
that there are no data sharing among the BSs and every BS has
perfect CSIs of all links. This is a scenario of asymmetric MIMO-IBC,
and the configuration is denoted as $\prod_{i=1}^{G}\left(M_i\times
\prod_{k=1}^{K_i}(N_{i_k},d_{i_k})\right)$. When $M_i=M$, $K_i=K$,
$N_{i_k}=N$ and $d_{i_k}=d$, $\forall i,k$, the system becomes a symmetric MIMO-IBC denoted as $\left(M\times
(N,d)^{K}\right)^{G}$.

Because both the IA with and without symbol extension will be
addressed, we define two terminologies to be used throughout the
paper.
\begin{definition}\label{Linear IA}
\emph{Linear IA} is the IA without any symbol extension or only with finite spatial extension \cite{Yetis2010}.
\end{definition}
\begin{definition}\label{asymptotic IA}
\emph{Asymptotic IA} is the IA with infinite time or frequency extension \cite{Jafar_3cell}.
\end{definition}

In this section, we study two necessary conditions of IA
feasibility.
\subsection{Proper Condition}
When the channels are generic (i.e., drawn from a continuous probability distribution), the proper condition is one necessary condition for linear IA feasibility, which has been obtained for asymmetric MIMO-IBC in \cite{TTTSP2013}. To find the difference of deriving the necessary conditions between symmetric and asymmetric systems, we first review the proper condition in brief, which is
\begin{align}\label{Eq:Proper_Condition_Asym}
&\sum\nolimits_{j:({i_k},j) \in {\cal I}} ( {{M_j} - {d_j}} ) {d_j} + \sum\nolimits_{{i_k}:({i_k},j) \in {\cal I}} ( {{N_{i_k}} - {d_{i_k}}}) {d_{i_k}} \nonumber\\
\ge &\sum\nolimits_{({i_k},j) \in {\cal I}} {{d_{i_k}}{d_j}} ,\forall {\cal I} \subseteq {\cal J}
\end{align}
where
\begin{align}\label{Eq:Definition_J}
{\cal J} \triangleq \{ ({i_k},j)|1 \le i \ne j \le G,1 \le k \le {K_i}\}
\end{align}
denotes the set of all MS-BS pairs that mutually interfering each other and ${\cal I}$ is an arbitrary subset of
$\mathcal{J}$.

From \eqref{Eq:Definition_J}, it is not hard to obtain that the set ${\cal J}$ has
\begin{align}\label{Eq:Num_J}
  L_{{\cal J}}=2^{|{\cal J}|}-1 = 2^{\sum_{j=1}^{G}\sum_{i=1,i\neq j}^{G}K_{i}}-1
\end{align}
nonempty subsets, where $|{\cal J}|$ is the cardinality of ${\cal
J}$. As a result, \eqref{Eq:Proper_Condition_Asym} includes $L_{{\cal J}}$ inequalities. By contrast, the proper condition for symmetric MIMO-IBC includes only one inequality, which is \cite{TTTSP2013}
\begin{align}\label{Eq:Proper_Condition_sym}
  M+N\geq(GK+1)d
\end{align}
and obtained by only considering ${\cal I}={\cal J}$ in \eqref{Eq:Proper_Condition_Asym}.

Then, we introduce another terminology.
\begin{definition}\label{Connection――pattern}
\emph{Connection pattern} is a graph that represents which BSs and users that mutually interfering each other in an arbitrarily connected MIMO-IBC.
\end{definition}

Each subset ${\cal I}$ corresponds to  a connection pattern. Take a two-cell MIMO-IBC where $K_1=2,K_2=1$ as an example, denoted as \emph{Ex. 1}. Since $L_{{\cal J}}=2^{K_1+K_2}-1=7$, there are seven connection patterns in Ex. 1 as shown in Fig.
\ref{fig:Diff_Connection}. According to \eqref{Eq:Definition_J}, we
know that the subset ${\cal I}={\cal J}$ corresponds to a \emph{fully connected pattern}, i.e.,
Pattern I, and the subset ${\cal I}\subset {\cal J}$ corresponds to a
\emph{partially connected pattern}, e.g., each one in Patterns II$\sim$VII.
\begin{figure}[htb!]
\centering
\includegraphics[width=1.0\linewidth]{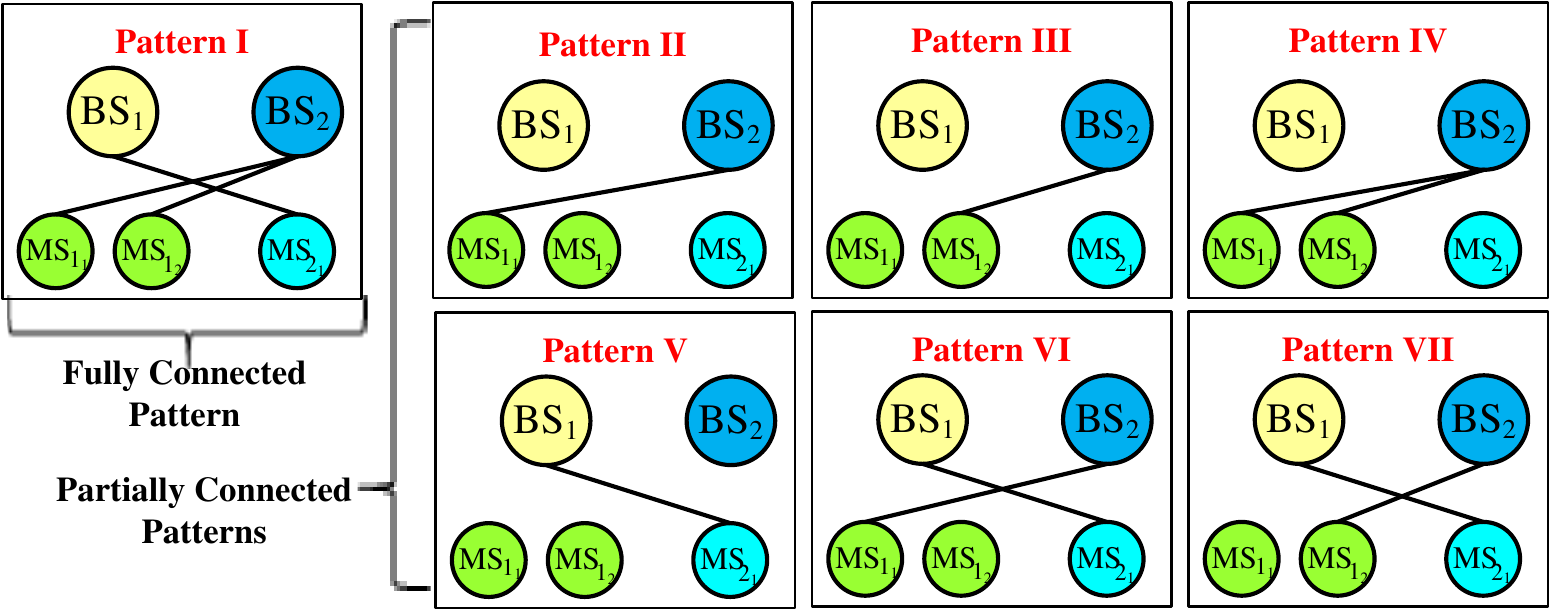}
\caption{Connection patterns in Ex. 1.} \label{fig:Diff_Connection}
\end{figure}

Comparing \eqref{Eq:Proper_Condition_Asym} and \eqref{Eq:Proper_Condition_sym}, we find that for the symmetric MIMO-IBC, it is enough to only consider the fully connected pattern. However, for the asymmetric MIMO-IBC, since all BSs or users have different number of antennas and suffer from different interference, it is necessary to consider all possible connection patterns rather than only consider the fully connected pattern.

%

\subsection{Irreducible Condition}
In \cite{TTTSP2013b}, another necessary condition for both linear IA and asymptotic IA feasibility has been found for symmetric MIMO-IBC, which leads to an information theoretic DoF upper-bound. It was called \emph{irreducible condition} in
\cite{TTTSP2013} since it can ensure to eliminate a kind of
irreducible inter-cell interference (ICI).

A usual way to derive the information theoretic DoF upper-bound is
introducing a \emph{genie} \cite{Jafar_3cell}. It is not easy to
find a useful and smart genie to provide the tightest possible
upper-bound. The analysis in \cite{TTTSP2013b} indicates that when
dividing the \emph{interference subspace} at each BS or user into
two linearly independent subspaces, i.e., \emph{resolvable subspace}
and \emph{irresolvable subspace}, we can construct a wise genie from
the irresolvable subspace. For some antenna configurations, there
may exist multiple irresolvable subspaces that constitute an
\emph{irresolvable subspace chain}, called \emph{subspace chain}.
Correspondingly, the genies in the subspace chain constitute a
\emph{genie chain}. To derive the DoF upper-bound (equivalently the
irreducible condition) for asymmetric MIMO-IBC, we need to first
investigate the subspace chain and genie chain. Although the principle of constructing the genie chain for
asymmetric MIMO-IBC is similar with that for symmetric MIMO-IBC, the
results in \cite{TTTSP2013b} cannot be extended in a straightforward
manner, which is shown in the forthcoming analysis.

\subsubsection{Subspace Chain}
The subspace chain depends on how BSs and users eliminate the interference cooperatively. When BSs first eliminate one part of ICIs and then return the remaining part to users, the exists a subspace chain that starts from the BS side, denoted as Chain A. Meanwhile, when users first eliminate one part of ICIs and then return the remaining to BSs, there also exists another subspace chain that starts from the user side,
denoted as Chain B. To derive the whole necessary condition, we need to consider Chains A and B simultaneously.

For easy understanding, we illustrate the subspace chain for asymmetric MIMO-IBC by a subspace chain of Ex. 1 in Fig. \ref{fig:Chain_A}.

\begin{figure}[htb!]
\centering
\includegraphics[width=1.0\linewidth]{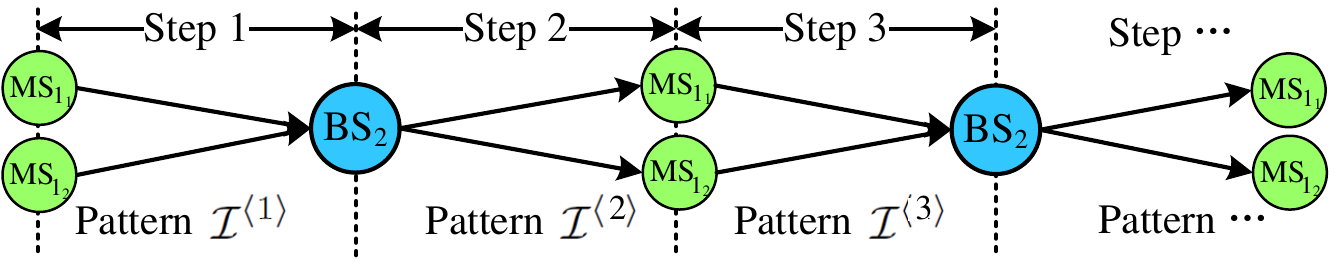}
\caption{A subspace chain of Ex. 1.} \label{fig:Chain_A}
\end{figure}

As shown in the figure, the connection pattern in each step is Pattern IV in Fig. \ref{fig:Diff_Connection}, i.e., BS$_2$ and MS$_{1_1}$, MS$_{1_2}$ mutually interfere each other. For BS$_2$, there are two interfering users and the ICIs generated by the two users
occupy $(N_{1_1}+N_{1_2})$-dimensional interference subspace at
BS$_2$. Since BS$_2$ has $M_{2}$ antennas, only
$\min\{(N_{1_1}+N_{1_2}),M_2\}$-dimensional subspace of the
interference subspace is resolvable at BS$_2$ and the remaining
$\left(N_{1_1}+N_{1_2}-M_{2}\right)^{+}$-dimensional subspace is
irresolvable. If $N_{1_1}+N_{1_2}-M_{2}\leq 0$, there does not exist
an irresolvable subspace so that the subspace chain stops. When
$N_{1_1}+N_{1_2}-M_{2}>0$, the irresolvable subspace at BS$_2$ is
nonempty and the subspace chain continues.

For MS$_{1_1}$ and MS$_{1_2}$, there is only one interfering BS and
the irresolvable ICIs occupy $(N_{1_1}+N_{1_2}-M_{2})$-dimensional
subspace. Since MS$_{1_1}$ and MS$_{1_1}$ have $N_{1_1}$ and
$N_{1_2}$ antennas to resolve $N_{1_1}$ and $N_{1_2}$-dimensional
subspace, the remaining $\left(N_{1_2}-M_{2}\right)^{+}$ and
$\left(N_{1_1}-M_{2}\right)^{+}$-dimensional subspaces are
irresolvable at MS$_{1_1}$ and MS$_{1_2}$, respectively.

We use $\mathcal{S}_{\alpha,j}^{\langle n \rangle}$ or $\mathcal{S}_{\alpha,i_k}^{\langle n \rangle}$ to denote BS$_j$ or MS$_{i_k}$'s irresolvable subspace in the $n$th step of Chain $\alpha,~\alpha=\mathrm{A},~\mathrm{B}$. For the subspace chain in Fig. \ref{fig:Chain_A}, the dimension of the irresolvable subspace at BS$_2$ in Step 1 is
\begin{align}\label{Eq:Dim_irr_space_BS}
  |\mathcal{S}_{\mathrm{A},2}^{\langle 1 \rangle}| =\left(N_{1_1}+N_{1_2}-M_{2}\right)^{+}
\end{align}
and the dimensions of the irresolvable subspaces at MS$_{1_1}$ and MS$_{1_2}$ in Step 2 are
\begin{subequations}
\begin{align}\label{Eq:Dim_irr_space_MS1}
  |\mathcal{S}_{\mathrm{A},1_1}^{\langle 2 \rangle}|& =\left(\mathcal{S}_{\mathrm{A},2}^{\langle 1 \rangle} -N_{1_1}\right)^{+}=\left(N_{1_2}-M_{2}\right)^{+}\\
  \label{Eq:Dim_irr_space_MS2}
  |\mathcal{S}_{\mathrm{A},1_2}^{\langle 2 \rangle}|& =\left(\mathcal{S}_{\mathrm{A},2}^{\langle 1 \rangle} -N_{1_2}\right)^{+}=\left(N_{1_1}-M_{2}\right)^{+}
\end{align}
\end{subequations}

If the considered connection pattern of the subspace chain in Fig. \ref{fig:Chain_A} is Pattern III but not Pattern IV, the dimension of
the irresolvable subspace at BS$_2$ in Step 1 becomes
$\mathcal{S}_{\mathrm{A},2}^{\langle 1 \rangle}
=\left(N_{1_2}-M_{2}\right)^{+}$. It indicates that the dimension of
irresolvable subspace also depends on the connection pattern.

To describe the all connection patterns in the subspace chain, we introduce the fourth terminology.
\begin{definition}
\emph{Connection pattern chain} is a chain constituted by the connection patterns
in different steps of a subspace chain, denoted as ${{\cal I}^{\langle 1
\rangle}}\leftrightarrow{{\cal I}^{\langle 2
\rangle}}\leftrightarrow\cdots$, where ${{\cal I}^{\langle n
\rangle}}$ is the $n$th pattern in the chain. If a connection pattern chain satisfies ${{\cal I}^{\langle n \rangle}}= {\cal I},~\forall n$, it is \emph{equally connected pattern chain} (ECPC), otherwise it is \emph{unequally connected pattern chain} (UCPC). If an ECPC satisfies ${{\cal I}^{\langle n \rangle}}= {\cal J},~\forall n$, it is a \emph{full ECPC}, otherwise it is a \emph{partial ECPC}.
\end{definition}

As shown in Fig. \ref{fig:Chain_A}, the connection pattern chain is Pattern IV$\leftrightarrow$Pattern IV$\leftrightarrow\cdots$, which is an ECPC.

Since all BSs or users have different number of antennas, the different BS or user nodes have different dimensions of irresolvable subspace. We know that in one step, the irresolvable subspace at some nodes may become empty but that at other nodes is nonempty.
Consequently, only the nonempty subspaces appear in the next step
of the subspace chain. As a result, the connection pattern for an
arbitrary asymmetric MIMO-IBC in the next step is always a subsect
of that in the current step, i.e.,
\begin{align}\label{Eq:CP_Chain_Propety}
{{\cal I}^{\langle n \rangle}} \supseteq {{\cal I}^{\langle n+1 \rangle}},~\forall n\geq1
\end{align}

From the above analysis, we know that to derive the necessary condition for asymmetric systems, all possible connection pattern chains satisfying \eqref{Eq:CP_Chain_Propety} need to be taken into account to construct the genie chain. In \cite{TTTSP2013b}, only the fully ECPU is investigated in deriving the genie chain so that the obtained genie chain is also a special case of the genie chain considered for asymmetric systems. In other word, by introducing a notion of connection pattern chain, we generalize the genie chain in \cite{TTTSP2013b} into asymmetric MIMO-IBC.

Following the similar way in \cite{TTTSP2013b} and considering the
difference mentioned above, for an arbitrary connection pattern
chain satisfying \eqref{Eq:CP_Chain_Propety}, we can obtain the
dimension of irresolvable subspace in each step for asymmetric
MIMO-IBC. To express the dimension of irresolvable subspace in a
unified way, we define $| {\cal S}_{\mathrm{A},i_k}^{\langle 0 \rangle}
|=| {\cal S}_{\mathrm{B},i_k}^{\langle -1\rangle}
|\triangleq N_{i_k},~| {\cal S}_{\mathrm{A},j}^{\langle -1\rangle}
|=| {\cal S}_{\mathrm{B},j}^{\langle 0 \rangle}
|\triangleq M_j$. Then, the dimension of irresolvable subspace in
each step of Chains A and B can be expressed as
\begin{subequations}\label{Eq:Chain_Subspace_A}
\begin{align}
&\left\{
\begin{array}{l}
| {\cal S}_{{\rm{A}},j}^{\langle 2n + 1\rangle } | = \left( \sum\nolimits_{i_k \in {\cal I}_{j}^{\langle 2n+1\rangle }} | {\cal S}_{{\rm{A}},{i_k}}^{\langle 2n\rangle } | - | {\cal S}_{{\rm{A}},j}^{\langle 2n - 1\rangle } |  \right)^ + \\
| {\cal S}_{{\rm{A}},i_k}^{\langle 2n + 2\rangle} | = \left( \sum\nolimits_{j \in {\cal I}_{i_k}^{\langle 2n+2\rangle }}
| {\cal S}_{{\rm{A}},j}^{\langle 2n+1\rangle } | - | {\cal S}_{{\rm{A}},i_k}^{\langle 2n \rangle } |  \right)^ +\\
\end{array}
\right.\\
&\left\{
\begin{array}{l}
| {\cal S}_{{\rm{B}},i_k}^{\langle 2n + 1\rangle } | = \left( \sum\nolimits_{j \in {\cal I}_{i_k}^{\langle 2n+1\rangle }} | {\cal S}_{{\rm{B}},{j}}^{\langle 2n\rangle } | - | {\cal S}_{{\rm{B}},i_k}^{\langle 2n - 1\rangle } |  \right)^ + \\
| {\cal S}_{{\rm{B}},j}^{\langle 2n + 2\rangle} | = \left( \sum\nolimits_{i_k \in {\cal I}_{j}^{\langle 2n+2\rangle }}
| {\cal S}_{{\rm{B}},i_k}^{\langle 2n+1\rangle } | - | {\cal S}_{{\rm{B}},j}^{\langle 2n \rangle } |  \right)^ +\\
\end{array}
\right.
\end{align}
\end{subequations}
$\forall 0\leq n \leq n_{\max}$, where ${\cal I}_{j}^{\langle n
\rangle} \triangleq \left\{ {i_k}|\left( {{i_k},j} \right) \in
{{\cal I}^{\langle n \rangle}} \right\}$ is the set of MSs' index
who are connected with tBS$_j$ in the connection pattern ${{\cal
I}^{\langle n \rangle}}$, ${\cal I}^{\langle n \rangle}_{i_k}
\triangleq \left\{ j|\left( {{i_k},j} \right) \in {{\cal I}^{\langle
n \rangle}} \right\}$ is the set of BSs' index who are connected
with MS$_{i_k}$ in ${{\cal I}^{\langle n \rangle}}$, $n_{\max}$
reflects the maximal length of subspace chain and satisfies $|{{\cal
S}^{\langle 2n_{\max}+1\rangle}_{\alpha,j}} | = 0,\forall
(i_k,j) \in {\cal I}^{\langle 2n_{\max}+1\rangle}$ or $|{{\cal
S}^{\langle 2n_{\max}+2\rangle}_{\alpha,i_k}} | = 0,\forall
(i_k,j) \in {\cal I}^{\langle 2n_{\max}+2\rangle}$, $\alpha=\mathrm{A},~\mathrm{B}$.

\subsubsection{Genie Chain}
From the dimension of the irresolvable subspace in
\eqref{Eq:Chain_Subspace_A}, we can determine the corresponding
dimension of the genie and then obtain the irreducible condition.

We use $\mathcal{G}_{\alpha,j}^{\langle n \rangle}$ or
$\mathcal{G}_{\alpha,i_k}^{\langle n \rangle}$ to denote the
introduced genie at BS$_j$ or MS$_{i_k}$ in the $n$th step of Chain
$\alpha,~\alpha=\mathrm{A},~\mathrm{B}$. In Step 1, if the irresolvable subspace at BS$_2$
$\mathcal{S}_{\mathrm{A},2}^{\langle 1 \rangle}$ in
\eqref{Eq:Dim_irr_space_BS} is nonempty, we can introduce a genie in
$\mathcal{S}_{\mathrm{A},2}^{\langle 1 \rangle}$ (denoted as
$\mathcal{G}_{\mathrm{A},2}^{\langle 1 \rangle}$) to help BS$_2$ to
resolve all ICIs in Step 1. Then, we have
\begin{align}\label{Eq:DoF_Genie_Bound}
  d_{1_1} + d_{1_2} \leq (M_2-d_{2}) +|\mathcal{G}_{\mathrm{A},2}^{\langle 1 \rangle}|
\end{align}

If $\mathcal{S}_{\mathrm{A},1_1}^{\langle 2 \rangle}$ and $\mathcal{S}_{\mathrm{A},1_2}^{\langle 2 \rangle}$ in  \eqref{Eq:Dim_irr_space_MS1} and \eqref{Eq:Dim_irr_space_MS2} are nonempty, we can introduce two genies in $\mathcal{S}_{\mathrm{A},1_1}^{\langle 2 \rangle}$ and $\mathcal{S}_{\mathrm{A},1_2}^{\langle 2 \rangle}$ (denoted as $\mathcal{G}_{\mathrm{A},1_1}^{\langle 2 \rangle}$ and $\mathcal{G}_{\mathrm{A},1_2}^{\langle 2 \rangle}$) to help MS$_{1_1}$ and MS$_{1_2}$ to resolve the remaining ICIs in Step 2. Then, we have
\begin{align}\label{Eq:DoF_Genie_Bound_MS1}
  |\mathcal{G}_{\mathrm{A},2}^{\langle 1 \rangle}| \leq d_{1_1} +|\mathcal{G}_{\mathrm{A},1_1}^{\langle 2 \rangle}|,~
  |\mathcal{G}_{\mathrm{A},2}^{\langle 1 \rangle}| \leq d_{1_2} +|\mathcal{G}_{\mathrm{A},1_2}^{\langle 2 \rangle}|
\end{align}


Following the same way, we can obtain the dimension of the genie for
the general asymmetric MIMO-IBC. To describe the genie's dimension
in a unified way, we define $| {{\cal G}^{\langle 0
\rangle}_{\mathrm{A},i_k}} | \triangleq  d_{i_k},~| {{\cal
G}^{\langle - 1\rangle}_{\mathrm{A},j}} | \triangleq  M_j - d_j$, $| {{\cal G}^{\langle 0
\rangle}_{\mathrm{B},j}} | \triangleq  d_{j},~| {{\cal
G}^{\langle - 1\rangle}_{\mathrm{B},i_k}} | \triangleq  N_{i_k} - d_{i_k}$.
Then, the dimension of the genie in each step satisfies
\begin{subequations}\label{Eq:Chain_Genie_A}
\begin{align}
&\left\{
\begin{array}{l}
 \sum\nolimits_{{i_k} \in {\cal I}^{\langle 2n+1 \rangle}_{j}} {| {{\cal G}^{ \langle 2n \rangle}_{\mathrm{A},i_k}} |}  \le | {{\cal G}^{\langle 2n - 1 \rangle}_{\mathrm{A},j}} | + | {{\cal G}^{\langle 2n + 1 \rangle}_{\mathrm{A},j}} | \\
\sum\nolimits_{j \in {\cal I}^{\langle 2n + 2 \rangle}_{{i_k}}} {| {{\cal G}^{\langle 2n + 1 \rangle}_{\mathrm{A},j}} |}  \le  | {{\cal G}^{\langle 2n \rangle}_{\mathrm{A},i_k}} | + | {{\cal G}^{\langle 2n + 2 \rangle}_{\mathrm{A},i_k}} |\\
\end{array}
\right.\\
&\left\{
\begin{array}{l}
\sum\nolimits_{{j} \in {\cal I}^{\langle 2n+1 \rangle}_{i_k}} {| {{\cal G}^{ \langle 2n \rangle}_{\mathrm{A},j}} |}  \le | {{\cal G}^{\langle 2n - 1 \rangle}_{\mathrm{A},i_k}} | + | {{\cal G}^{\langle 2n + 1 \rangle}_{\mathrm{A},i_k}} | \\
\sum\nolimits_{i_k \in {\cal I}^{\langle 2n + 2 \rangle}_{{j}}} {| {{\cal G}^{\langle 2n + 1 \rangle}_{\mathrm{A},i_k}} |}  \le  | {{\cal G}^{\langle 2n \rangle}_{\mathrm{A},j}} | + | {{\cal G}^{\langle 2n + 2 \rangle}_{\mathrm{A},j}} |\\
\end{array}
\right.
\end{align}
\end{subequations}
$\forall 0\leq n \leq n_{\max}$.

Moreover, since the genie in each step lie in its corresponding
irresolvable subspace, the dimension of  the genie does not exceed
that of the irresolvable subspace, i.e.,
\begin{subequations}\label{Eq:Chain_Genie_A1}
\begin{align}
&| {{\cal G}^{\langle 2n - 1 \rangle}_{\mathrm{A},j}} | \le | {{\cal S}^{\langle 2n - 1 \rangle}_{\mathrm{A},j}} |,~| {{\cal G}^{\langle 2n \rangle}_{\mathrm{A},i_k}} | \le | {{\cal S}^{\langle 2n \rangle}_{\mathrm{A},i_k}} |\\
&| {{\cal G}^{\langle 2n - 1 \rangle}_{\mathrm{B},i_k}} | \le | {{\cal S}^{\langle 2n - 1 \rangle}_{\mathrm{B},i_k}} |,~| {{\cal G}^{\langle 2n \rangle}_{\mathrm{B},j}} | \le | {{\cal S}^{\langle 2n \rangle}_{\mathrm{B},j}} |
\end{align}
\end{subequations}
$\forall 0\leq n \leq n_{\max}$.

Combining the inequalities in \eqref{Eq:Chain_Genie_A} and
\eqref{Eq:Chain_Genie_A1}, we obtain the whole \emph{irreducible condition} for asymmetric
MIMO-IBC.

\begin{remark}
Since $| {{\cal S}^{\langle 0 \rangle}_{\mathrm{A},i_k}} | =
N_{i_k}$ and $| {{\cal G}^{\langle 0 \rangle}_{\mathrm{A},i_k}} | =
d_{i_k}$, $| {{\cal S}^{\langle 0 \rangle}_{\mathrm{B},j}} | =
N_{j}$ and $| {{\cal G}^{\langle 0 \rangle}_{\mathrm{B},j}} | =
d_{j}$, from \eqref{Eq:Chain_Genie_A1} we have $d_{i_k}\leq
N_{i_k}$ and $d_j \leq M_j$. As a result, the irreducible
condition contains the condition that ensuring each BS or MS with
enough antennas to convey the desired signals.
\end{remark}

\subsection{ Comparison of Two Necessary Conditions}

The irreducible condition is one necessary condition for both linear IA and asymmetric IA feasibility, while the proper condition is
necessary for linear IA feasibility but no necessary for asymmetric IA feasibility.
Therefore, the information theoretical outer-bound of DoF region can
be derived from the irreducible condition, while the outer-bound of
DoF region achieved by linear IA needs to be derived from both the
irreducible condition and the proper condition.

\begin{remark}
If the configuration of a MIMO-IBC satisfies the proper condition but
not the reducible condition, the system is proper but infeasible for
linear IA.
\end{remark}

Since it is very difficult to obtain a general result for arbitrary
configurations of asymmetric MIMO-IBC, in the following we use an
example to show how to obtain the information theoretical outer-bound
from the irreducible condition. The outer-bound achieved by linear IA can be
obtained in a similar way.

Consider Ex.1 again, if the configuration satisfies
$|\mathcal{S}_{\mathrm{A},2}^{\langle 1 \rangle}|>0$,
$|\mathcal{S}_{\mathrm{A},1_1}^{\langle 2 \rangle}|=0$ and
$|\mathcal{S}_{\mathrm{A},1_2}^{\langle 2 \rangle}|=0$, i.e.,
$N_{1_1}+N_{1_2}> M_{2} \geq \max\{N_{1_1},N_{1_2}\}$, from
\eqref{Eq:Dim_irr_space_BS} we can introduce a genie to help BS$_2$,
but cannot introduce any genie to help MS$_{1_1}$ and MS$_{1_2}$.
Substituting into \eqref{Eq:Chain_Genie_A} and
\eqref{Eq:Chain_Genie_A1}, the reducible condition is
\begin{align}\label{Eq:Chain_Genie_A_Ex}
\left\{
\begin{array}{l}
d_{1_1} \leq N_{1_1}, d_{1_2} \leq N_{1_2}, d_{2} \leq N_{2} \\
d_{1_1}+d_{1_2} + d_{2} \leq M_2 + |\mathcal{G}_{\mathrm{A},2}^{\langle 1 \rangle}| \\
|\mathcal{G}_{\mathrm{A},2}^{\langle 1 \rangle}| \leq d_{1_1},~  |\mathcal{G}_{\mathrm{A},2}^{\langle 1 \rangle}| \leq d_{1_2}\\
|\mathcal{G}_{\mathrm{A},2}^{\langle 1 \rangle}| \leq |\mathcal{S}_{\mathrm{A},2}^{\langle 1 \rangle}|=N_{1_1}+N_{1_2} - M_2\\
\end{array}
\right.
\end{align}
By solving \eqref{Eq:Chain_Genie_A_Ex}, the outer-bound of DoF region is obtained as
\begin{align}
\left\{
\begin{array}{l}
  d_{1_1} \leq N_{1_1}, d_{1_2} \leq N_{1_2}, d_{2} \leq N_{2} \\
  d_{1_1}+ d_{2} \leq M_2, d_{1_2}+ d_{2} \leq M_2 \\
  d_{1_1}+ d_{1_2}+ d_{2} \leq N_{1_1}+N_{1_2}
\end{array}
\right.
\end{align}

In the chain shown in Fig. \ref{fig:Chain_A}, we only consider the connection pattern chain where ${{\cal I}^{\langle n \rangle}}$ is Pattern IV, $\forall n$. If we consider other connection pattern chain, we can obtain the other outer-bound of DoF region in a similar way.


\section{DoF per user for Two-cell Asymmetric MIMO-IBC}
Because the proper condition and the irreducible condition for general asymmetric
MIMO-IBC are too complicated for analysis, we consider a class of special
configurations to obtain the closed-form expression of the DoF. We
consider that all users have the same number of data streams, i.e.,
$d_{i_k}=d,~\forall k,i$, then the outer-bound of DoF region can be
characterized by the upper-bound of DoF per user. Besides, we
consider that all users in one cell have the same number of receive
antennas, i.e., $N_{i_k}=N_i,~\forall k$, then the DoF of these
users can be analyzed in a unified way.

For such a special configuration, we can derive the closed-form DoF
upper-bound per user for two-cell asymmetric MIMO-IBC and prove that
it is the upper-bound of DoF per user for $G$-cell asymmetric
MIMO-IBC. Note that such a two-cell asymmetric MIMO-IBC represents a
typical heterogeneous network where a macro-cell and a micro-cell
interfere each other, so that its DoF results can provide useful
insights into interference management in heterogeneous networks. In
the following, we investigate the two-cell MIMO-IBC, denoted as
$\prod_{i=1}^{2}(M_i\times (N_i,d)^{K_i})$.

\subsection{Information Theoretic Maximal DoF}
To understand the potential of the two-cell asymmetric MIMO-IBC, we
first investigate the information theoretic maximal DoF per user.
\begin{theorem}[Information Theoretic Maximal DoF]\label{Theorem:DoF_Upperbound}
For a  two-cell MIMO-IBC $\prod_{i=1}^{2}(M_i\times (N_i,d)^{K_i})$, the information theoretic maximal DoF per user is
\begin{align}\label{Eq:DoF_Info_ALL}
\min_{i\neq j}\left\{d^{\mathrm{Info}}\left(M_j,N_i,K_j,K_i\right)\right\}
\end{align}
where
\begin{align*}
  d^{\mathrm{Info}}\left(M_j,N_i,K_j,K_i\right)
\triangleq
   \left\{
  \begin{array}{ll}
    d^{\mathrm{Decom}}(M_j,N_i,K_j),&\mathrm{Region~I}\\
    d^{\mathrm{Quan}}(M_j,N_i,K_j,K_i),&\mathrm{Region~II}\\
   \end{array}
  \right.
\end{align*}
\begin{align}\label{Eq:Decompositon DoF}
  d^{\mathrm{Decom}}(M_j,N_i,K_j)\triangleq\frac{{{M_j}{N_i}}}{{{M_j} + {K_j}{N_i}}}
\end{align}
\begin{align}\label{Eq:Quantity DoF}
&d^{\mathrm{Quan}}\left( {M_j},{N_i},K_j,k \right)\\
\triangleq &\left\{
\begin{array}{ll}
\min \left\{ {\frac{{{M_j}}}{{{K_j} + C_n^{\mathrm{A}}\left( k \right)}},\frac{{{N_i}}}{{1 + \frac{{{K_j}}}{{C_{n - 1}^{\mathrm{A}}\left( k \right)}}}}} \right\},&
\forall C_n^{\mathrm{A}}\left( k \right) \le \frac{M_j}{N_i} < C_{n - 1}^{\mathrm{A}}\left( k \right)\\
\min \left\{ {\frac{{{M_j}}}{{{K_j} + C_{n - 1}^{\mathrm{B}}\left( k \right)}},\frac{{{N_i}}}{{1 + \frac{{{K_j}}}{{C_n^{\mathrm{B}}\left( k \right)}}}}} \right\},&
\forall C_{n - 1}^{\mathrm{B}}\left( k \right) < \frac{M_j}{N_i} \le C_n^{\mathrm{B}}\left( k \right)
\end{array} \right.\nonumber
\end{align}
\begin{align*}
\begin{array}{ll}
  \mathrm{Region~I:} & C_{\infty} ^{\mathrm{B}}\left( {{K_i}} \right) <  \frac{M_j}{N_i}  <  C_{\infty} ^{\mathrm{A}}\left( {{K_i}} \right),~~\forall K_i\geq 4\\
  \mathrm{Region~II:} &
  \left\{ \begin{array}{ll}
            \mathrm{Arbitrary~}M_j,N_i,  & \forall K_i\leq 3 \\
            \frac{M_j}{N_i}\geq  C_{\infty} ^{\mathrm{A}}\left( {{K_i}} \right) ~\mathrm{or}~\frac{M_j}{N_i} \leq C_{\infty} ^{\mathrm{B}}\left( {{K_i}} \right), & \forall K_i\geq 4
          \end{array}
  \right.
\end{array}
\end{align*}
$C_n^{\mathrm{A}}\left( k \right) \triangleq {k} - {k}/{{C_{n - 1}^{\mathrm{A}}\left( k \right)}},~
C_n^{\mathrm{B}}\left( k \right) \triangleq {k}/({k} - C_{n - 1}^{\mathrm{B}}( k))$, $C_{0}^{\mathrm{A}}=\infty$, $C_{0}^{\mathrm{B}}=0$, and $C_{\infty} ^{\mathrm{\alpha}}(K_i ) \triangleq \lim \limits_{n \to \infty } C_n^{\mathrm{\alpha}}(K_i )$, $\alpha =\mathrm{A},~\mathrm{B}$.
\end{theorem}

When the system reduces to a symmetric two-cell MIMO-IBC, the DoF bounds in \eqref{Eq:Decompositon DoF} and \eqref{Eq:Quantity DoF} reduce to the \emph{decomposition DoF bound} and \emph{quantity DoF bound} in \cite{TTTSP2013b}, respectively. Therefore, we also call \eqref{Eq:Decompositon DoF} and \eqref{Eq:Quantity DoF} as \emph{decomposition DoF bound} and \emph{quantity DoF bound}.

Due to the lack of space, we only provide the skeleton of proofs for
all
theorems.

\begin{proof}[Proof Skeleton]
We first prove that \eqref{Eq:DoF_Info_ALL} is the information theoretic DoF upper-bound and then prove that it is achievable.

For the conclusion that the quantity DoF bound is the information
theoretic upper-bound in Region II, a rigorous proof was provided in
\cite{TTTSP2013b} for symmetric MIMO-IBC. Using similar derivations,
the quantity DoF bound can be derived from the irreducible condition
where the full ECPC is considered, which is the information
theoretic DoF upper-bound.

For the conclusion that the decomposition DoF bound is the
information theoretic upper-bound in Region I, there is no rigorous
proof in existing studies. From the analysis in last section, we
know that the subspace chain depends on the connection pattern
chain. This means that for different connection pattern chains, we
can obtain different DoF upper-bounds from the irreducible
condition. Moreover, we show that when considering the full
ECPC, the DoF upper-bound derived from the irreducible condition is
only applicable for the antenna configurations in Region II.
However, if considering the partial ECPCs or UCPUs, the derived DoF
upper-bound is applicable for the antenna configurations in Region I. With more derivations, we find that for
some antenna configurations in Region I, the derived DoF upper-bound
is equal to the decomposition DoF bound. It means that the
decomposition DoF bound is the DoF upper-bound for these antenna
configurations.

To help understand the proof, we illustrate the result by an example
in Fig. \ref{fig:IA_Feasible3}. From the boundary of feasible
region, the information theoretic DoF upper-bound is shown. From the
boundaries of infeasible regions, the DoF upper-bounds derived from
the irreducible condition are shown. When considering the full ECPC,
the DoF upper-bound is obtained only for the antenna configurations
in Region II. However, when considering a given partial ECPC, the
DoF upper-bound is obtained for antenna configurations in Region I,
which is equal to the decomposition DoF bound for some
configurations.
\begin{figure}[htb!]
\centering
\includegraphics[width=0.9\linewidth]{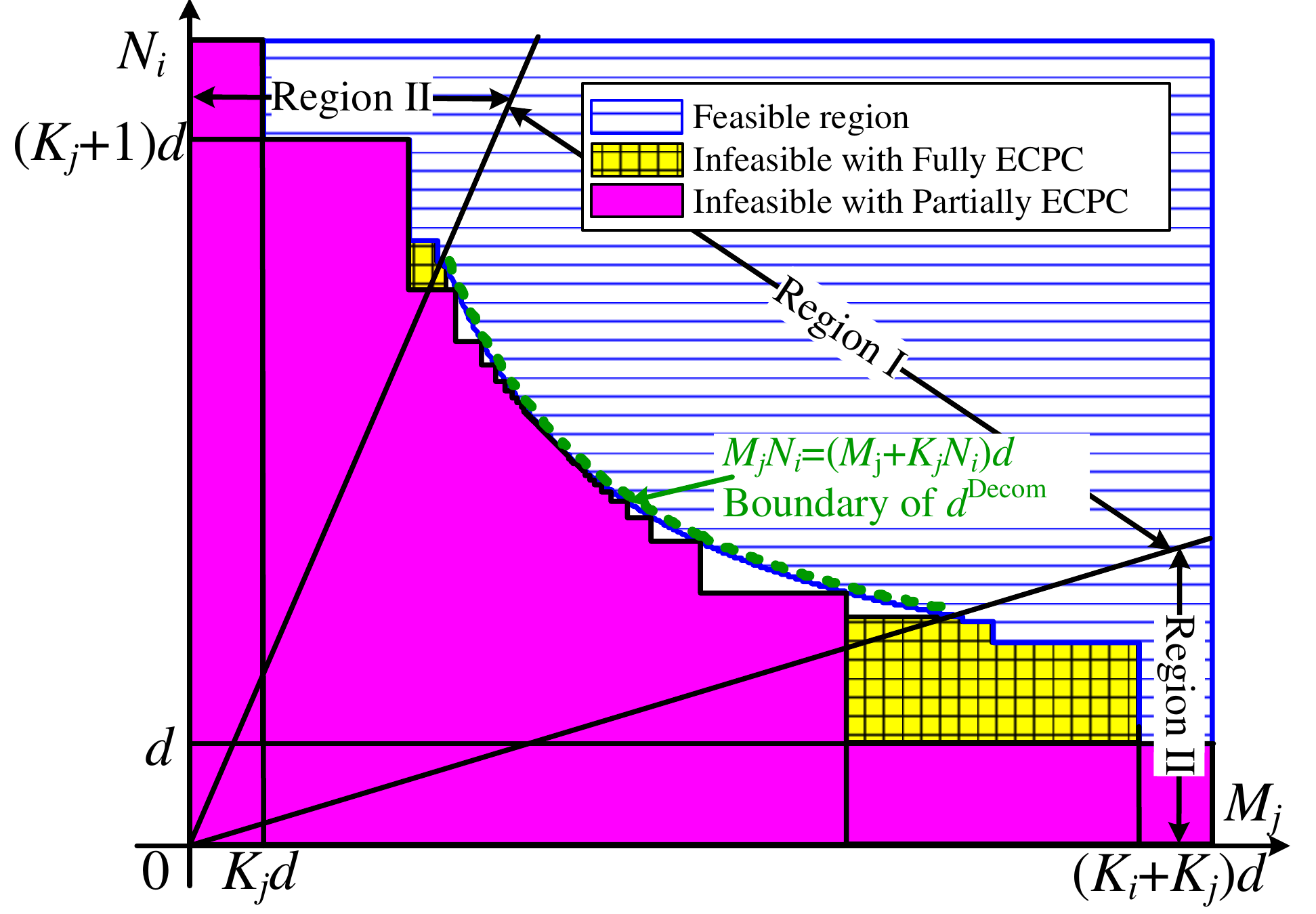}
\caption{Feasible and infeasible regions for a two-cell MIMO-IBC $\prod_{i=1}^{2}(M_i\times (N_i,d)^{K_i})$.}
\label{fig:IA_Feasible3}
\end{figure}

According to the properties of the
generalized continue fraction and generalized Fibonacci
sequence-pairs in \cite{Continued_Fractions1980}, we prove that for
an arbitrary antenna configuration in Region I, there always exists
a connected pattern chain ensuring that the DoF upper-bound obtained
from the reducible condition is equal to the decomposition DoF
bound. As a result, this proves that the decomposition DoF bound is
the information theoretic DoF upper-bound in Region I.

In the following, we prove the DoF upper-bound is achievable.
Following similar derivations in \cite{TTTSP2013b}, we can show that
for the antennas configuration when both $M_1/N_2$ and $M_2/N_1$
fall in Region II, the DoF upper-bound can be achieved by linear
IA and the closed-form solution of linear IA exists. Obviously, the
DoF can also be achieved by asymmetric IA but the infinite
extension is not necessary.  For other configurations, the DoF
upper-bound can only be achieved by asymptotic IA but not by
linear IA.
\end{proof}

\begin{remark}
For the considered two-cell MIMO-IBC, the information theoretic DoF
upper-bound for arbitrary configurations can be obtained from the irreducible condition
and can always be achieved by asymptotic IA. Therefore, the irreducible
condition is the sufficient and necessary condition for asymptotic
IA feasibility.
\end{remark}

In \cite{TTTSP2013b}, the irreducible condition is obtained as the sufficient and necessary condition of asymptotic IA feasibility for the antenna configurations in Region II and what is the condition for the antenna configurations in Region I is still unknown. In this study, we prove that the irreducible condition obtained from the generalized genie chain is the sufficient and necessary condition of asymptotic IA feasibility for arbitrary antenna configurations.

\subsection{Maximal Achievable DoF of Linear IA}
Considering that asymptotic IA requires infinite time/frequency
extension, which is not feasible for practical systems, this
motivates us to find the maximal DoF per user achieved by linear IA.
\begin{theorem}[Maximal  DoF achieved by Linear IA]\label{Theorem:Linear_DoF_Upperbound}
For a two-cell MIMO-IBC $\prod_{i=1}^{2}(M_i\times (N_i,d)^{K_i})$, the maximal DoF per user achieved by linear IA is
\begin{align}\label{Eq:Proper_DoF_ALL}
\min_{i\neq j} \left\{d^{\mathrm{Prop}}\left(M_j,N_i,K_j,K_i  \right),d^{\mathrm{Info}}\left(M_j,N_i,K_j,K_i \right)
  \right\}
\end{align}
where
\begin{align}\label{Eq:Proper DoF}
  d^{\mathrm{Prop}}\left({{M_j},{N_i}},K_j,K_i \right) = \frac{{{K_j}{M_j} + {K_i}{N_i}}}{{K_j^2 + {K_j}{K_i} + {K_i}}}
\end{align}
\end{theorem}

\begin{proof}[Proof Skeleton]
We first prove that \eqref{Eq:Proper_DoF_ALL} is the DoF upper-bound and then prove that it is achievable.

After some tedious but regular manipulations, we can show that in
the considered two-cell MIMO-IBC, the proper condition in
\eqref{Eq:Proper_Condition_Asym} becomes
\begin{align}\label{Eq:Proper_Condition_Asym_G2}
{\left( {{M_j} - {K_jd}} \right)} {K_jd} + {\left( {{N_{i}} - {d}} \right)} {K_id} \ge
K_jK_id^2，~\forall i\neq j
\end{align}
From \eqref{Eq:Proper_Condition_Asym_G2}, we can obtain $d \leq d^{\mathrm{Info}}\left({{M_j},{N_i}},K_j,K_i\right)$ in \eqref{Eq:Proper DoF} directly.

Since \eqref{Eq:Proper DoF} is one DoF upper-bound obtained from the proper condition, it is called \emph{proper DoF bound}. Moreover, $d^{\mathrm{Info}}\left({{M_j},{N_i}},K_j,K_i\right)$ is another DoF upper-bound obtained from the irreducible condition, which is proved in Theorem \ref{Theorem:DoF_Upperbound}. Consequently, \eqref{Eq:Proper_DoF_ALL} is the upper-bound of the DoF per user achieved by linear IA.

Following the similar analysis in \cite{TTTSP2013b}, we can show
that if the proper DoF bound is not higher than the information
theoretic maximal DoF, the DoF can be achieved by a closed-form linear
IA, otherwise, there exist at least one feasible solution for linear
IA following a similar proof as in \cite{TTTSP2013}.
\end{proof}

\begin{remark}
In the considered two-cell MIMO-IBC, the maximal DoF per user
achieved by linear IA is obtained from the proper condition and the
irreducible condition simultaneously. Therefore, the combination of
the proper condition and the irreducible condition is the sufficient
and necessary condition for linear IA feasibility.
\end{remark}

\begin{remark}
When the proper DoF bound is achievable, \eqref{Eq:Proper DoF}
implies that to achieve the desired number of data streams, if one
BS increases or reduces the transmit antennas, the MSs in another
cell can reduce or should increase the receive antennas.
Specifically, let $L={K}_{j}/{K}_{i}$ where $K_j>K_i$. If the number
of transmit antennas at BS$_j$ is $M_{j}\pm\Delta$, the number of
receive antenna at MSs in cell $i$ should be $N_{i}\mp L \Delta$. As
a result, when a cell that supports more users (or equivalently that
needs to transmit more data streams) has redundant antennas to help
eliminate ICI, the overall antenna resource in the network can be
reduced. This suggests that a heterogeneous network with different
downlink traffics and antenna resources among multiple cells is more
spectrally efficient than a homogeneous network, if we allow the
more powerful cell (such as a macro-cell whose BS is equipped with
more antennas) to help remove the ICI for the resource-limited cell
(such as a micro-cell).
\end{remark}

\begin{remark}
Since a two-cell MIMO-IBC can be treated as a partially connected case of a $G$-cell MIMO-IBC, the necessary condition for the two-cell MIMO-IBC must be the necessary condition for the $G$-cell MIMO-IBC. Consequently, the DoF per user upper-bounds in Theorems \ref{Theorem:DoF_Upperbound} and \ref{Theorem:Linear_DoF_Upperbound} are the DoF per user upper-bounds for $G$-cell MIMO-IBC $\prod_{i=1}^{G}(M_i\times (N_i,d)^{K_i})$ with asymptotic and linear IA, respectively.
\end{remark}


%
%
%
\section{Conclusion}
In this paper, we analyzed the DoF of the asymmetric MIMO-IBC. By generalizing the genie chain considered in \cite{TTTSP2013b}, we found that irreducible condition is
necessary for both linear IA and asymptotic IA feasibility. From the irreducible condition, we derived the information theoretic DoF outer-bound for arbitrary $G$-cell MIMO-IBC and the information theoretic maximal DoF per user for a special class of two-cell MIMO-IBC with the antenna configurations in both Regions I and II. By contrast, the reducible condition derived in \cite{TTTSP2013b} only leads to the information theoretic maximal DoF for symmetric MIMO-IBC with the antenna configurations in Region II.
By combining the proper condition and irreducible condition, we obtained the DoF outer-bound for arbitrary MIMO-IBC achieved by linear IA and the maximal achievable DoF per user for the special two-cell MIMO-IBC. By comparing the information theoretic maximal DoF and the maximal DoF achieved by linear IA, we showed when the linear IA can achieve the information theoretic maximal DoF.

\bibliographystyle{IEEEtran}


\end{document}